\documentclass[conference,letterpaper]{IEEEtran}
\usepackage{amssymb, cite}
\usepackage{bm,bbm}
\usepackage{mathtools}
\usepackage{multirow,color,amsfonts}
\usepackage{tabulary}
\usepackage{subfigure}
\usepackage{graphicx}
\usepackage{setspace}
\usepackage{enumerate}
\usepackage[]{algorithm2e}
\usepackage{comment}

\usepackage[utf8]{inputenc} 
\usepackage[T1]{fontenc}
\usepackage{ifthen}
\usepackage{cite}

\usepackage[hyphens]{url}
\usepackage[hidelinks]{hyperref}
\hypersetup{breaklinks=true}
\usepackage{amsthm}
\usepackage{amsmath}

\newtheorem{defi}{Definition}

\newtheorem{prop}{Proposition}

\newtheorem{cor}{Corollary}

\newtheorem{lem}{Lemma}
\newtheorem{ex}{Example}

\newcommand{\CGE}{{graph entropy}}
\newcommand{\FCGE}{{fractional graph entropy}}

\newcommand{\FCN}{{fractional chromatic number}}

\newcommand{\FCE}{{fractional chromatic entropy}}
\newcommand{\CE}{{chromatic entropy}}

\newcommand{\HG}{{H_{G_{X_1}}(X_1\vert X_2)}}
\newcommand{\HGfrac}{{H^{f}_{G_{X_1}}(X_1\vert X_2)}}

\newcommand{\HGchi}{{H^{\chi}_{G_{X_1}}(X_1\vert X_2)}}
\newcommand{\HGchifrac}{{H^{\chi_f}_{G_{X_1}}(X_1\vert X_2)}}

\newcommand{\Graph}{G_{X_1}}
\newcommand{\twoPowerGraph}{G^2_{{\bf X_1}}}

\newcommand{\nPowerGraph}{G^n_{{\bf X_1}}}

\newcommand{\HGchipowertwo}{{H^{\chi}_{G^2_{{\bf X_1}}}({\bf X_1}\vert {\bf X_2})}}
\newcommand{\HGchipowertwofrac}{{H^{\chi_f}_{G^2_{{\bf X_1}}}({\bf X_1}\vert {\bf X_2})}}
\newcommand{\HGchipowern}{{H^{\chi}_{G^n_{{\bf X_1}}}({\bf X_1}\vert {\bf X_2})}}
\newcommand{\HGchipowerk}{{H^{\chi}_{G^k_{{\bf X_1}}}({\bf X_1}\vert {\bf X_2})}}
\newcommand{\HGchipowerkminusone}{{H^{\chi}_{G^{k-1}_{{\bf X_1}}}({\bf X_1}\vert {\bf X_2})}}
\newcommand{\HGchipowernfrac}{{H^{\chi_f}_{G^n_{{\bf X_1}}}({\bf X_1}\vert {\bf X_2})}}
\newcommand{\HGchipowerkfrac}{{H^{\chi_f}_{G^k_{{\bf X_1}}}({\bf X_1}\vert {\bf X_2})}}
\newcommand{\HGchipowerkminusonefrac}{{H^{\chi_f}_{G^{k-1}_{{\bf X_1}}}({\bf X_1}\vert {\bf X_2})}}

\newcommand{\coloring}{c_{{G_{X_1}}}(X_1)}
\newcommand{\coloringx}{c_{{G_{X_1}}}}
\newcommand{\coloringpowern}{c_{{G^n_{{\bf X_1}}}}({\bf X_1})}
\newcommand{\coloringpowernplusone}{c_{{G^{n+1}_{{\bf X_1}}}}({\bf X_1})}
\newcommand{\coloringpowernx}{c_{{G^n_{{\bf X_1}}}}}
\newcommand{\coloringpowertwox}{c_{{G^2_{{\bf X_1}}}}}

\newcommand{\coloringf}{c^f_{{G_{X_1}}}(X_1)}
\newcommand{\coloringxf}{c^f_{{G_{X_1}}}}
\newcommand{\coloringpowernf}{c^f_{{G^n_{{\bf X_1}}}}({\bf X_1})}

\newcommand{\coloringpowernplusonef}{c^f_{{G^{n+1}_{{\bf X_1}}}}({\bf X_1})}
\newcommand{\coloringpowernxf}{c^f_{{G^n_{{\bf X_1}}}}}

\begin{document}
\title{%A Refinement on 
Fractional Graph Coloring for \\Functional Compression with Side Information%\vspace{-0.2cm}
} 
%I think this approach is lossless in the limit as b goes to infinity: H(f(X,Y)|Y)= \HGchifrac (the {\FCE})

 %%% Single author, or several authors with same affiliation:
 \author{
   \IEEEauthorblockN{Derya Malak}
   \IEEEauthorblockA{Communication Systems Department,   EURECOM\\
               %Sophia Antipolis, 06904 FRANCE\\%SophiaTech -  
               derya.malak@eurecom.fr}
}

\maketitle

\begin{abstract}
We describe a rational approach to reduce the computational and communication complexities of lossless point-to-point compression for computation with side information. The traditional method relies on building a characteristic graph with vertices representing the source symbols and with edges that assign a source symbol to a collection of independent sets to be distinguished for the exact recovery of the function. %realizing the computation task.  %captures the mappings from the source to the function outputs. %that represents the source mappings onto the function outputs, i.e., 
Our approach uses fractional coloring for a b-fold coloring of characteristic graphs to provide a linear programming relaxation to the traditional coloring method and achieves coding at a fine-grained granularity.
We derive the fundamental lower bound for compression, given by the fractional characteristic graph entropy, through generalizing the notion of K\"orner's graph entropy. We demonstrate the coding gains of fractional coloring over traditional coloring via a computation example. We conjecture that the integrality gap between fractional coloring and traditional coloring approaches the smallest b that attains the fractional chromatic number to losslessly represent the independent sets for a given characteristic graph, up to a linear scaling which is a function of the fractional chromatic number.%in the limit, as the length of the data stream tends to infinity.}
\end{abstract}

%%%%%%%%%%%%%%%%%%%%%%%%%%%%%%%%%%%%%%%%%%%%%%%
\section{Introduction}
\label{background}

We consider the problem of point-to-point compression for computing a function with decoder side information. %, as shown in Fig. \ref{fig:side_info}. 
Traditionally, this problem is referred to as K\"orner's graph coloring problem \cite{korner1973coding}. 
%\ab{consider moving the figure to top of the next column, also consider removing the sentence: `` We generalize this method via a fractional coloring approach and demonstrate the potential rate savings.'' and focus on the intro to the problem.}
This problem stems from source coding (compression), which has been the subject of extensive study in information theory dating back to the seminal work of Shannon \cite{shannon1948mathematical}, and its many extensions. 
%In compression, the goal of this encoding is parsimony, i.e., we wish to encode sources %such that they take the minimal amount of space, or equivalently, 
%such that the representation is shortest (in terms of the number of bits). 

%%%%%%%%%%%%%%%%%%%%%%%%%%%%%%%%%%%%%%%
\subsection{Coding for Compression}
\label{subsection:coding}

In the traditional compression approach, for a point to point compression of a source variable $X_1$ drawn from distribution $P_{X_1}$, the source coding theorem, as demonstrated by Shannon \cite{shannon1948mathematical}, states that in the limit, as the length of a stream of independent and identically-distributed (i.i.d.) random variable data tends to infinity, the best rate of compression of $X_1$ (quantified in the average number of bits per symbol) is the Shannon entropy of the source, $H(X_1)=\mathbb{E}[-\log P_{X_1}(X_1)]$.
%"information", "surprise", or "uncertainty" inherent to the variable's possible outcomes

A natural extension of the source coding theorem for networked settings is the problem of {\em distributed compression}. 
The problem of distributed lossless compression dates back to the seminal work of Slepian and Wolf \cite{slepian1973noiseless}, instantiated by random binning of the typical source sequences \cite{cover1975proof}. For concreteness, consider two random variables $X_1$ and $X_2$, jointly distributed according to $P_{X_1,X_2}$.
Given two sequences ${\bf X}_1^n=(X_{11},X_{12},\dots, X_{1n})$ and ${\bf X}_2^n=(X_{21},X_{22},\dots, X_{2n})$ drawn i.i.d. from $P_{X_1,X_2}$, Slepian-Wolf Theorem gives a theoretical bound for the lossless coding rate of distributed coding \cite{slepian1973noiseless}: To recover a joint source $({\bf X}^n_1,{\bf X}^n_2)$ drawn from $P_{X_1,X_2}$ at a receiver that has access to side information ${\bf X}^n_2$, it is both necessary and sufficient to encode the source ${\bf X}^n_1$ %s  $X_1$ and $X_2$ 
up to the rate $H(X_1|X_2)$  \cite{slepian1973noiseless}. 
Slepian-Wolf problem is a special case of the general distributed function compression problem. %that we consider in Fig.~\ref{fig:side_info}, 
%where $f({\bf X_1}, {\bf X_2}) = ({\bf X_1}, {\bf X_2}).$

%\ab{I think you want to define $(X_1^n, X_2^n)$ as an iid sequence of length $n$ drawn from $P_{X_1, X_2}$}
%\ab{technically you need a sequence of length $n$ and achievability only works as $n \to \infty$} \derya{ok, yes, I know, is it better now?}
%Note that the encoding is done in a truly distributed way, that is, no communication or coordination is necessary between the encoders. The Slepian-Wolf theorem shows that making use of the correlation, at the expense of vanishing error probability for long sequences, allows a  compression rate equal to that with cooperaton. 

Practical schemes for Slepian-Wolf compression have been proposed by several authors, including \cite{PR031, sartipi2008distributed, zhang2018generalized, 1705002}. The generalization of the  distributed compression scheme of Slepian-Wolf to trees and to networks beyond depth one uses random linear coding, as shown by Ho {\em et al.} in \cite{1705002}.
%As long as the total rate of $X_1$ and $X_2$ is larger than their joint entropy $H(X_1,X_2)$ and none of the sources is encoded with a rate smaller than its conditional entropy, distributed coding can achieve arbitrarily small error probability for long sequences. While the pro
Distributed communication has also been considered in \cite{1705002}, and by Ahlswede {\em et al.} \cite{AhlCaiLiYeu2000} and Yeung in \cite{yeung2002first} for multicasting under general network settings via random linear network coding. %Ahlswede {\em et al.} \cite{AhlCaiLiYeu2000} showed that with network coding, as the symbol size approaches infinity, a source can multicast information at a rate approaching the smallest minimum cut between the source and any receiver. 

\begin{comment}
\begin{figure}[t!]
\centering
\includegraphics[width=0.7\columnwidth]{SideInfo.pdf}
%\vspace{-0.2cm}
\caption{Computing problem with side information.}
\label{fig:side_info}
\end{figure} 
\end{comment}

%%%%%%%%%%%%%%%%%%%%%%%%%%%%%%%%%%%%%%%
\subsection{Coding for Functional Compression}
\label{subsection:coding}

%The traditional objective is to represent the data itself, and not a task derived from the data. The latter problem, i.e.,  
Distributed compression of source variables for the purpose of computing a deterministic function across a network, is referred to as {\em distributed functional compression}. 
%The state of the art  currently states  theoretically that function-oriented representation can be far more efficient than generic data-centric compression. 
To that end, since the pioneering work of Slepian-Wolf \cite{slepian1973noiseless}, different techniques have been explored, e.g., computation with {\em decoder side information} and {\em functional distortion criterion} in Wyner-Ziv settings~\cite{WynZiv1976}, compression for multiple descriptions of functions \cite{gamal1982achievable}, special functions such as addition \cite{korner1979encode}, and multiplication with side information \cite{watanabe2013rate}.
% 
%While entropy and conditional entropy are the quantities that govern the fundamental limits of compression when the goal is to recover the data itself, the case of 
Function-oriented recovery, or compression %for computation 
of $f(X_1,X_2)$ %(this is the two-user special case while more general multi-variate functions can be considered), 
is better understood through the lens of characteristic graph-entropy $H_{G_{X_1}}(X_1)$ \cite{korner1973coding}, which quantifies the minimum number of bits required to represent a function of random variables. %The characteristic graph $G_{X_1}$ is constructed according the distribution of the source variable $X_1$, and the specific structure of $f(X_1,X_2)$. 
This notion was initially devised by K\"orner \cite{korner1973coding} for point-to-point compression.

The zero-error side information problem and the rate regions for the functional compression problem have been investigated by Witsenhausen \cite{witsenhausen1976zero}, along with the formal introduction of the characteristic graph of $X_1$, $X_2$ and $f$,  %Associated with the source pair $(X_1, X_2)$ is a characteristic graph $G$. Its vertex set is X and two distinct vertices $x$ and $x'$$ are connected if they are confusable, i.e., if there is a $y$ such that $p(x,y),p(x',y)>0$ and $f(x,y)\neq f(x',y)$. 
by Orlitsky-Roche \cite{OR01} when one source is fully available at the receiver via conditional graph entropy $H_{G_{X_1}}(X_1|X_2)$, and for restricted and unrestricted inputs by Alon-Orlitsky \cite{alon1996source}. The problem of distributed lossless compression has been studied by Doshi {\em et al.} \cite{doshi2006graph}, and also extended to tree %s and general 
networks (Feizi-M{\'e}dard \cite{feizi2014network} and Doshi {\em et al.} \cite{DSME10}) via generalizing  distributed compression (Slepian-Wolf \cite{slepian1973noiseless}) to {\em distributed functional compression}. %The Slepian-Wolf theorem is the natural scenario where the function $f(X_1,X_2)$ is the identity function. %such that the sources $X_1$ and $X_2$ can be asymptotically compressed up to the rate $ H_{G_{X_1}}(X_1|X_2)$ when $X_2$ is available at the receiver. %along with generalizations to other distributed settings:
%The potential savings of functional compression over source compression can be characterized by the gap between the inner bound for compression, as determined by the {\em distributed communication rate region}, and the outer bound determined by the {\em characteristic graph entropy}.

{\bf \em The requirement for structured coding for computing.} %K\"orner-Marton \cite{korner1979encode} have demonstrated the need for structured codes in certain cases, but no general framework is available. Special cases of distributed computing of functions over communication networks include the open problems of {\em distributed functional compression} \cite{OR01} and {\em communicating correlated sources over a multiple access channel} \cite{korner1979encode}, with little progress since their introduction in the 1980s. These instances demonstrate the requirement for structured codes.
For distributed compression of a general function, 
%a code optimized for communication, i.e., 
trimming of (independently encoded) Slepian–Wolf partitions may not be feasible. 
%In existing achievable schemes, encoder finds mappings from data onto its equivalence classes and uses graph coloring-based encoding for these equivalence classes, e.g., \cite{OR01}, \cite{AO96}, \cite{korner1979encode}, and \cite{Korner1973}. We note that graph coloring-based approaches rely on NP-hard schemes, and impose necessary conditions on the functions for the valid colorings which are hard or impossible to satisfy.  However, the general problem of nonlinear function encoding is open. 
In other instances, good computation codes may only achieve marginal gains in computing capacity over separation-based codes (see e.g., \cite{nazer2007computation} and \cite{feizi2014network}) at the expense of a significant computation {\em burden on the encoders and decoders}. 
There also exist approaches to compression of graphical data in sparse scenarios, e.g., \cite{delgosha2020universal,li2019unified,magner2018lossless}, %\cite{delgosha2018distributed,delgosha2020universal,DelAnaISIT2020,Delgosha2019notion,li2019unified,anantharam2010information,choi2012compression,magner2018lossless}, 
which may not capture {\em computations of general functions}. %{\em lack the functional compression and the structural information aspects} in general. 
Hence, %{\em optimal compression} codes may impose a significant computation {\em burden on the encoders and decoders, and %realizing low-complexity computing in networks requires a fresh vision, different from the na\"ive linear principles.
designing efficient function-oriented %low-complexity 
codebooks requires a different vision to alleviate the redundancy of data-oriented encoding.

{\bf \em The need for constructive techniques to compression.}
For existing graph entropy-based approaches, %a thorough discussion 
we refer the reader to  \cite{doshi2006graph,dehmer2011history, alon1996source, OR01, korner1973coding, FES04, AnupRao2010lectures} and the references therein. The proofs of functional compression are by means of coloring characteristics graphs of functions but provide no natural constructive approach %that matches the bounds of the existence results
%formulated as an existence theorem without a constructive method for finding the object whose existence it proves %https://en.wikipedia.org/wiki/Entropy_compression
to instantiate functional compression. %There are not, to our knowledge, design algorithms to exploit those theoretical possibilities. 

%\ab{In contributions you want to give a clear sense of what goes beyond [31]. A sloppy reader may feel that most techniques are borrowed} %\ab{Ok this is great! I think you want to state it as such.}
%%%%%%%%%%%%%%%%%%%%%%%%%%%%%%%%%%%%%%%%%%%%%%%
Our contributions in this paper are summarized as follows:

\begin{itemize}
    \item In Sect. \ref{section:setup}, we provide %a primer on graph coloring-based coding techniques for functional compression. 
    a %linear programming 
    relaxation to the traditional graph coloring approach (NP-hard) for functional compression. This is possible through fractional coloring, by generalizing traditional graph coloring that assigns one color per vertex. %by generalizing K\"orner's graph coloring method that assigns one color per vertex to $b$-fold colorings. %originally devised for the zero-error point-to-point functional compression problem.
    An $a:b$ fractional graph coloring assigns $b$ colors out of a total of $a$ available colors to each vertex of a graph such that adjacent vertices have disjoint colors. %and has a reduced computational complexity over the traditional approach.
    %which is known as a $a:b$ coloring \cite{scheinerman2011fractional}.
    %The analysis for fractional coloring is possible via borrowing tools from \cite{scheinerman2011fractional}. 
    %The $b$-fold coloring problem is expressed as a linear program (i.e., accepts a polynomial-time solution) unlike the original coloring problem which is an integer programming problem (NP-complete).
    
   %Uncomment this in arxiv
    \item In Sect. \ref{section:codes}, we provide a binary code construction for fractional coloring, and determine the cost of encoding.

    \item We introduce the concept of the {\em fractional chromatic entropy} and characterize it in Prop \ref{FCE_characteristic_graph} exploiting %the chain rule for the entropy rate and Han's inequality 
    Han's theorem 
    \cite{te1978nonnegative}, which states that the average entropy decreases monotonically in the size of the subset.   
    In Prop. \ref{FCE_characteristic_graph_n_limit}, we derive the analytical expression for the {\FCGE}, $\HGfrac$, that gives the achievable rate for functional compression with side information. %\ab{consider breaking this into two bullets unless the two Props convey the same message and I failed to see it}
    \item Sect. \ref{section:FCE} presents our main results providing a %extending the existing fundamental 
    rate bound on functional compression using fractional coloring. Exploiting several %interesting 
    properties of {\FCN}, we state that $\HGfrac$ lower bounds $\HG$ (Lemma \ref{FCGEvsCGE}). %\ab{What is the takeaway from Prop. 3 (Lemma 1)? Is it that $\HGfrac$ characterizes the fundamental limit of this problem and $\HG$ does not?} %\derya{according to (\ref{chromatic}), yes to both questions.} 
    This new notion provides a refinement in coloring such that on average less colors are spent and the communication complexity is reduced. 
    \item Sect. \ref{section:coding_gains} provides several bounds on the integrality gap between fractional and traditional colorings. To contrast the potential rate savings of our approach over traditional coloring, %the notion of K\"orner's graph entropy for functional compression
we show that the integrality gap %, which is lower bounded by $1$ (Lemma \ref{FCGEvsCGE}), 
is a monotonically increasing function of the source sequence length $n$ (Lemma \ref{IG_as_function_n}), and approximate the integrality gap, and %conjecture that (Conj. \ref{conj_IG_lower_bound}) 
observe that (Prop. \ref{prop_IG_lower_bound})
it scales linearly with $b$, %in the limit of large $n$, 
given a valid $b$-fold coloring.  
%\ab{perhaps move around and define b first}
Hence, our approach yields {\em a reduced communication complexity} by a factor of $b$ (up to a linear scaling), in the %\footnote{Communication complexity quantifies the minimum 
number of communication or exchanged bits \cite{rao2020communication}, \cite{andrew1979some},  %\cite{abelson1980lower}
%} 
versus the traditional approach. 

\end{itemize}

%%%%%%%%%%%%%%%%%%%%%%%%%%%%%%%%%%%%%%%%%%%%%%%

%Outline: 1. Problem setup
%2.the code construction (given a fractional coloring)
%3. our construction satisfies the setup
%4. rest of the analysis.
%%%%%%%%%%%%%%%%%%%%%%%%%%%%%%%%%%%%%%%
\section{Problem Setup}
\label{section:setup}

We consider the problem of lossless distributed functional compression with side information introduced in \cite{OR01} (via generalizing \cite{WynZiv1976} using a characteristic graph approach). The encoder has source $X_1$. The decoder has source $X_2$, which is not accessible at the encoder side. Given two statistically dependent i.i.d. finite-alphabet random sequences $X_1^n$ and $X_2^n$, our goal is to give a theoretical bound for the lossless coding rate to encode $X_1$ for computing a function $f(X_1,\,X_2)$ to achieve arbitrarily small error probability for long sequences. Both the encoder and the decoder knows the function. 
In \cite{OR01} and its extensions, this problem, for the zero-error setting, has been tackled using a traditional coloring approach and it has been proven that a zero-error compression up to a rate $ H_{G_{X_1}}(X_1|X_2)$ is possible when $X_2$ is available at the receiver.

\begin{figure}[t!]
%\vspace{-0.3cm}
\centering
\includegraphics[width=0.75\columnwidth]{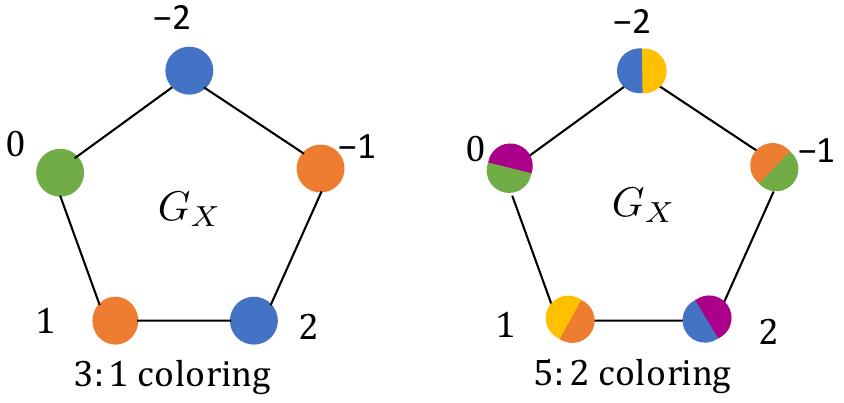}
\caption{(Left) %The characteristic graph 
$\Graph$ for Example \ref{uniform_example} with $|\mathcal{X}_1|=5$, where different colors on connected vertices indicate that those should be distinguished. %(Middle) A fractional $6:2$ coloring. 
(Right) A fractional $5:2$ coloring, which achieves $\chi_f$ \cite{scheinerman2011fractional}. %\ab{I think the definition of a:b coloring should appear before the example.} \derya{it is now under contributions.}
%\ab{do you need the $6:1$ color? I thought the goal here was to show that if you start increasing $b,$ you can find a coloring with an $a$ such that a/b is smaller.}\derya{Here, $\chi_b=2.5$ is optimal. If you try to have $b=3$, then since you cannot reuse each color more than twice, $a>7$. Hence, $a/b=8/3>2.5$. For $b=4$, you need at least $10$ colors.}
}%\ab{I meant do you need 6:2?}\derya{yes}
%\vspace{-0.3cm}
\label{fig:fractional_coloring}
\end{figure} 

%To explain our code construction we next introduce several definitions pertaining to traditional graph coloring (Sect. \ref{section:traditional_coloring}) and fractional graph coloring (Sect. \ref{section:fractional_coloring}).

%%%%%%%%%%%%%%%%%%%%%%%%%%%%%%%%%%%%%%%
\subsection{Traditional Coloring of Characteristic Graphs}% and Entropy of Characteristic Graphs
\label{section:traditional_coloring}

Let $G_{X_1}$ be the characteristic graph the encoder builds for computing the function $f(X_1,\,X_2)$, determined as function of $X_1$, $X_2$, and $f$. The characteristic graph is denoted by $G_{X_1}=(V_{G_{X_1}},\,E_{G_{X_1}})$, where $V_{G_{X_1}}=\mathcal{X}_1$ and an edge $(x_1^1, x_1^2)\in E_{G_{X_1}}$ if and only if there exists a $x_2^1 \in \mathcal{X}_2$ such that $p(x^1_1, x^1_2)\cdot p(x^2_1, x^1_2) > 0 $ and $f(x^1_1, x^1_2)\neq f(x^2_1, x^1_2)$. We assign different codes (colors) to connected vertices, which corresponds a graph coloring. Vertices that are not connected to each other can be assigned to the same or different colors. In this paper, we only consider vertex colorings. %Not all possible colorings of $\Graph$ are valid. 
A valid coloring of a graph $\Graph$ is such that each vertex of $\Graph$ is assigned a color such that adjacent vertices receive disjoint colors.

To better motivate our approach and demonstrate the achievable description lengths, we illustrate the relevance of characteristic graph in compression via the following example.

%Sample complexity bounds for dictionary learning from vector-and tensor-valued data \cite{shakeri2019sample}, 
\begin{ex}\label{uniform_example}{\bf A characteristic graph and its entropy.} %\cite{feizi2014network}
Random variables $X_1$ and $X_2$ are over the %same 
alphabet $\mathcal{X}=\{-2,\,-1,\,0,\,1,\,2\}$. The joint distribution %$P_{X_1,\,X_2}$ 
has ordered entries:
\begin{align}
\label{JointProbabilityTable}
P_{X_1,\,X_2}=\begin{bmatrix}
0.1 & 0.1 & 0 & 0 & 0\\
0.1 & 0 & 0 & 0 & 0.1\\
0 & 0.1 & 0.1 & 0 & 0\\
0 & 0 & 0.1 & 0.1 & 0\\
0 & 0 & 0 & 0.1 & 0.1
\end{bmatrix},   
\end{align}
where $X_1$ is uniformly distributed, and $f(X_1,X_2)=X_1+X_2$ such that $\Graph$ denotes the characteristic graph the encoder builds, where $V_{G_{X_1}}=\mathcal{X}$, and $E_{G_{X_1}}=\{(-2,-1),(-2,0),(0,1),(1,2),(2,-1)\}$. %, and $\{(-2,2),(-1,1),(0,2)\}\notin E_{G_{X_1}}$ for any given $x_2\in\mathcal{X}$. 
A valid coloring of $G_{X_1}$, denoted by $\coloring$ and shown in Fig. \ref{fig:fractional_coloring} (Left), %\ab{I am lost at the description of the example. What is the joint distribution of $(X_1, X_2)$?}
%\derya{we do not need it. The graph is based on a given $X_2=x_2$. It is a conditional coloring given $X_2$ and its entropy $H_{G_{X_1}}(X_1|X_2)$.} %\ab{are you assuming that regardless of $x_2,$ the coloring is the same? I think that puts a strong constraint on the joint distribution, which at least needs to be explicitly called out.} %\derya{In building $\Graph$ if we need to distinguish 2 vertices for any given value of $X_2$, then there is an edge. (This is redundant but how it is being implemented). However, it is not regardless of $x_2$. If $x_2$ is given, $f(X_1,x_2)$ is a subset of all possible outcomes. Then we need to decide the independent sets accordingly.} %\ab{I think a concrete example would help.} 
has a distribution $P(c_1)=P(c_2)=0.4$ and $P(c_3)=0.2$ over $\{c_1,\,c_2,\,c_3\}$. This yields an entropy %\footnote{If $P_{X_1,X_2}\neq P_{X_1}\cdot P_{X_2}$, then the coloring $\coloringx$ is a function of $X_2$ and its entropy can be determined as $H(\coloring \vert X_2)=\sum\nolimits_{x_2\in\mathcal{X}_2}P_{X_2}(x_2)\cdot H(\coloring \vert x_2)$ \cite{korner1973coding}. We later detail this dependence in (\ref{chromatic_vs_characteristic}).} 
$H(\coloring)\approx 1.52$.

Next, we encode a source sequence with length two, %random variable
${\bf X}_1^2=(X_{11},X_{12})$% \ab{is this a sequence of length 2? If so, please say that this is the notational choice}
, which can take $25$ values $\{(-2,-2),\,(-2,-1),\,(-2,0),\hdots, (2,2)\}$. %\ab{are these tuples? If so, what happened to $(-2, -2)$?} \derya{good catch. I fixed them now. Look it helps having a 2nd eye!} 
To construct the characteristic graph for ${\bf X}_1^2$, i.e., the second power graph $\twoPowerGraph$, we connect two vertices if at least one of coordinates are connected in $\Graph$. %$\twoPowerGraph$, the second power graph of $\Graph$, is illustrated in Fig. \ref{chargraph_and_powergraph} (Right). 
It is possible to color $\twoPowerGraph$  using $8$ colors. The entropy of this coloring satisfies $\frac{1}{2}H(\coloringpowertwox)\approx 1.48 < H(\coloringx)\approx 1.52<H(X_1)\approx 2.32$.
\end{ex}

\begin{figure}[t!]
%\vspace{-0.3cm}
\centering
\includegraphics[width=0.75\columnwidth]{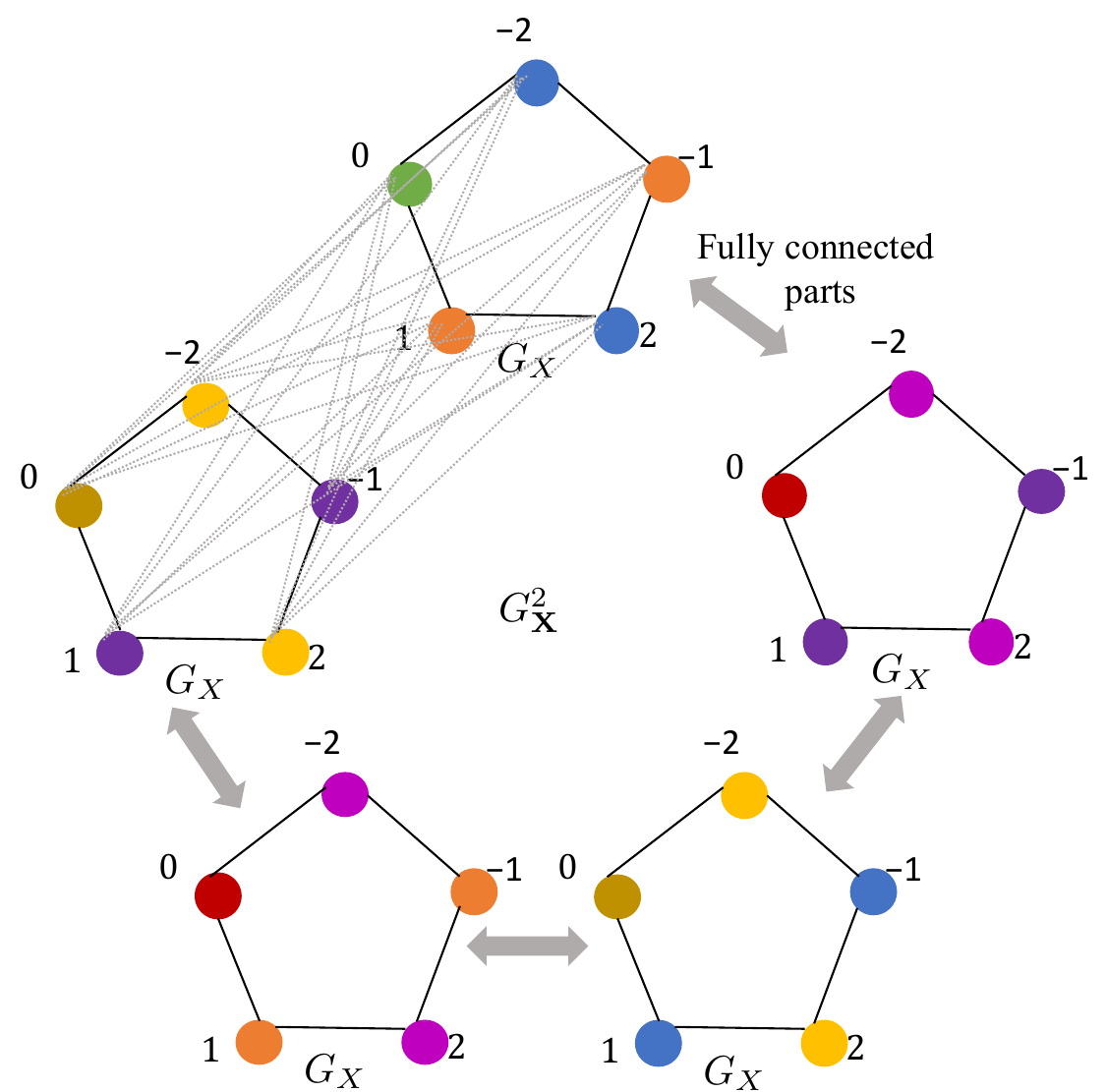}
\caption{The second power graph $\twoPowerGraph$ for Example \ref{uniform_example}, where $a=8$ is the minimum number of colors for which an $a:1$ coloring is possible. %\ab{is $a=8$ the smallest $a$ for which an $a:1$ coloring is possible?}
}
%\vspace{-0.3cm}
\label{fig:2nd_Power_Graph}
\end{figure} 

The chromatic number $\chi(\Graph)$ of a graph $\Graph$ is the minimum number of colours needed to colour the vertices in such a way that no two adjacent vertices have the same colour.

%\ab{consider making this a formal definition}
\begin{defi}\label{chromatic_entropy}
(Chromatic entropy \cite{alon1996source}.)
The {\CE} of a graph $\Graph$ is defined as 
\begin{align}
\label{chromatic}
\HGchi= \min_{\coloringx} \mathcal{H}^{\chi}(\coloring) \ , 
\end{align} 
where $\mathcal{H}^{\chi}(\coloring)=\{H(\coloring):\, \coloring \mbox{ is a valid coloring of } G_{X_1} \vert\, X_2\}$ is the set of chromatic entropies over the set of %all 
valid colorings of $\Graph$. %given $X_2$. 
\end{defi}

%Definition 12 \cite{feizi2014network}. A vertex coloring of a graph is a function $\coloring : V_{x_1} \to \mathbb{N} $ of a graph $\Graph =(V_{X_1} , E_{X_1} )$$ such that $(x_1^1, x_1^2) \in E_{X_1}$ implies $\coloringx(x_1^1)\neq \coloringx(x_1^2)$.

%Definition 14. 
Let $\nPowerGraph=(V_{X_1}^n , E_{X_1}^n )$ be the n-th power of a graph $\Graph$ such that $V_{X_1}^n = \mathcal{X}_1^n$ and $({\bf x}_1^1,{\bf x}_1^2) \in E_{X_1}^n$, where ${\bf x}_1^1=(x_{11}^1,x_{12}^1,\dots,x_{1n}^1)$ and similarly for ${\bf x}_1^2$, when there exists at least one coordinate $i\in\{1,2,\dots,n\}$ such that $(x_{1i}^1 ,x_{1i}^2 ) \in E_{X_1}$. We denote a coloring of $\nPowerGraph$ by $\coloringpowern$. 
K\"orner showed in \cite{korner1973coding} that, in the limit of large $n$, the {\CE} and the {\CGE} are related as
\begin{align}
\label{chromatic_vs_characteristic}
\HG%&=\lim\limits_{n\to \infty}\frac{1}{n}\HGchipowern\nonumber\\
&=\lim\limits_{n\to \infty}\frac{1}{n} \min\limits_{\coloringpowernx} H(\coloringpowern\vert {\bf X}_2) \ . 
\end{align}
The entropy of %the coloring of $\nPowerGraph$ 
$\coloringpowern$ characterizes %specifies 
the minimal %representation or 
description length needed to reconstruct with fidelity $f(X_1,X_2)$ \cite{korner1973coding}. The %degenerate 
case of the identity function yields %corresponds to having 
a complete %characteristic 
graph. %We focus on the lossless reconstruction of $f(X_1,\,X_2)$, as detailed in Sect. \ref{section:FCE}.

%\vspace{-0.1cm}
%%%%%%%%%%%%%%%%%%%%%%%%%%%%%%%%%%%%%%%
\subsection{Fractional Coloring of Characteristic Graphs}
\label{section:fractional_coloring}

%In traditional graph coloring, each vertex is assigned a color and the vertices connected by edges, i.e., the adjacent vertices, have different colors. 
Fractional graph coloring is a natural extension of traditional coloring, where each vertex is assigned a set of colors and the adjacent vertices have disjoint sets. 
Traditional graph coloring problems may not be amenable to a linear programming approach. 
Solving (\ref{chromatic_vs_characteristic}) is equivalent to determining a coloring random variable which minimizes the entropy. However, finding the minimum entropy coloring of $\Graph$ is an NP-hard problem \cite{cardinal2008tight}. 
To solve the coloring problem losslessly in polynomial time, we exploit the fractional coloring relaxation.

%We next detail some definitions and several interesting properties for fractional coloring of graphs. 

\begin{figure}[t!]
%\vspace{-0.3cm}
\centering
\includegraphics[width=0.75\columnwidth]{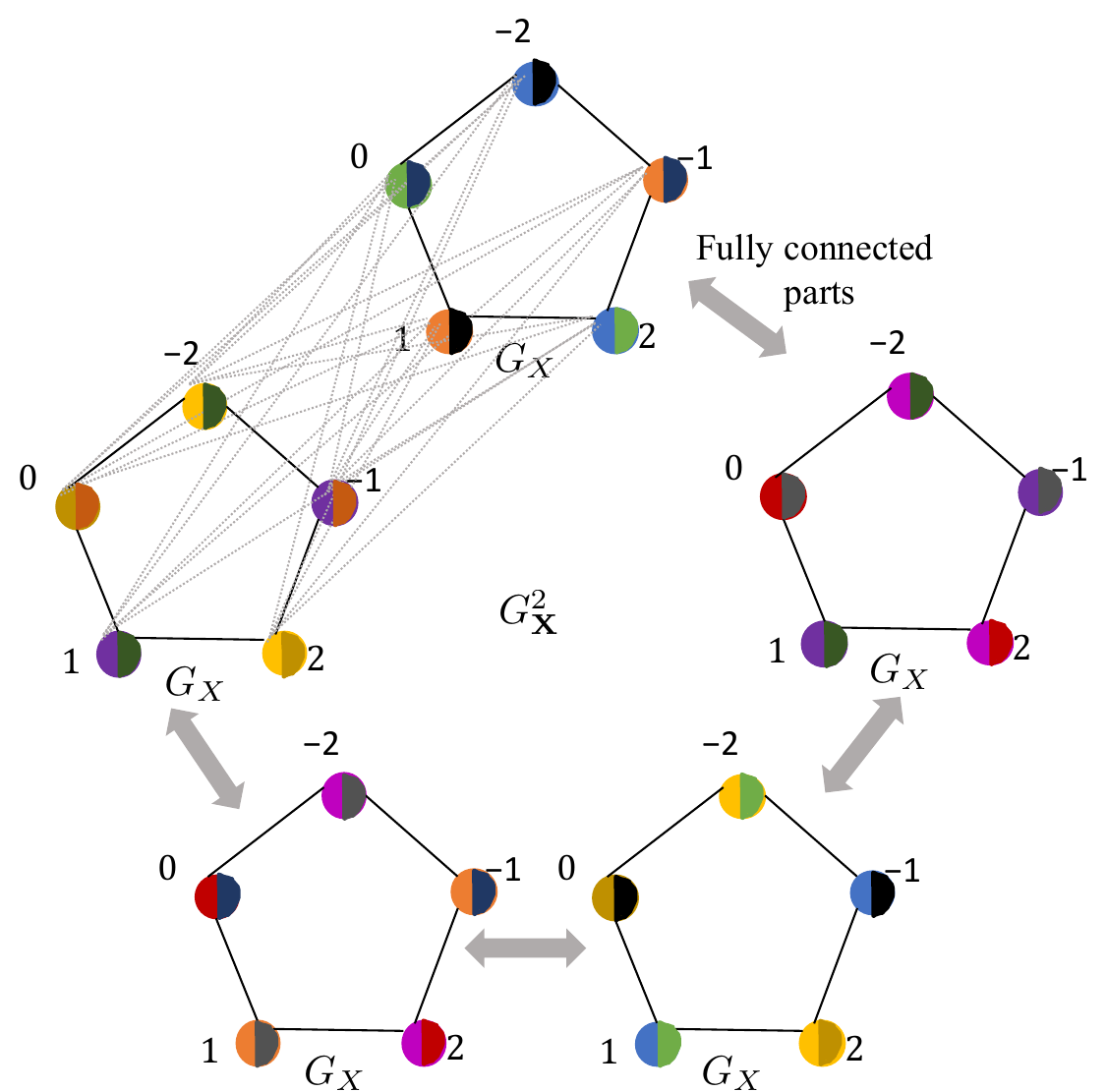}
\caption{A fractional coloring of  $\twoPowerGraph$ for Example \ref{uniform_example}, where $a=13$ is the minimum number of colors for which an $a:2$ coloring exists. %\ab{is $a=13$ the smallest $a$ for which an $a:2$ coloring is possible? If so, does that mean that you can turn this into a more efficient code? Can you show what the code looks like?}%\derya{Yes, it is, with $b=2$. Yes, I show the savings in the next column. You mean the code as a concatenation of half-colors? } 
%\ab{In the end of the day, as a source coding person I don't care about colors; they are tools. I care about the actual code in 0s and 1s. I think you'd want to make those connections crystal clear} 
} 
%\vspace{-0.3cm}
\label{fig:2nd_Power_Graph_Fractional}
\end{figure}

\begin{defi}[Scheinerman and Ullman~\cite{scheinerman2011fractional}]
A valid $b$-fold coloring is an assignment of sets of size $b$ to vertices such that adjacent vertices receive disjoint sets of colors. A valid $a:b$ coloring is a valid $b$-fold coloring out of $a$ available colors. The notation $\chi_b(G)$ represents the $b$-fold chromatic number of graph $G$ that is the least $a$ such that an $a:b$ coloring exists. %\cite{scheinerman2011fractional}%\ab{consider putting the reference at top}\derya{good idea}
\end{defi}

The chromatic number $\chi(G)$ is subadditive, i.e., $\chi _{a+b}(G)\leq \chi _{a}(G)+\chi _{b}(G)$. If $g: \mathbb{Z}^+ \to \mathbb{R}$ is subadditive and $g(b) \geq 0$ for all $b$, then from the sub-additivity lemma \cite{scheinerman2011fractional}, the limit $\lim\limits_{b\to\infty} \frac{g(b)}{b}$ exists and is equal to the infimum of $\frac{g(b)}{b}\,\, (b \in \mathbb{Z}^+)$.

\begin{defi}
The {\FCN} is defined as
\begin{align}
\label{fractional_graph_coloring}
\chi_f(G):=\liminf\limits_{b\to\infty}\left\{\frac{\chi_b(G)}{b}\right\}=\inf\limits_{b} \frac{\chi_b(G)}{b}\ , 
\end{align}
where the existence of this limit follows from the sub-additivity of $b$-fold colorings, and the sub-additivity lemma \cite{scheinerman2011fractional}. 
\end{defi}
%\vspace{-0.2cm}
%Fractional coloring is a relaxation of traditional coloring. 
From a probabilistic perspective, $\chi_f(G)$ represents the smallest $k$ for which there is a %probability 
distribution over the independent sets (an independent set is a set of vertices in a graph, no two of which are adjacent) of $G$ such that for each vertex $v$, given an independent set $I$ drawn from the distribution, $\mathbb{P}(v\in I)\geq\frac{1}{k}$.
Let $\mathcal{I}(G)$ be the set of all independent sets of $G$, and $\mathcal{I}(G,x)$ be the set of all those independent sets which include vertex $x$, and $x_I\in\mathbb{R}^+$ %be a nonnegative real variable 
for each independent set $I$. Then, the {\FCN} $\chi _{f}(G)$ can be obtained as a solution of the following linear program \cite{scheinerman2011fractional}:  
%\begin{comment}
\begin{align}
\label{chi_f_optimal}
\chi _{f}(G) \!=\! \min\limits_{\forall x} \Big\{ \sum _{I\in {\mathcal {I}}(G)}\!x_{I} : \sum_{I\in\mathcal{I}(G,x)} \!x_I \geq 1 , \,\, x_I\geq 0 \Big\}\ .    
\end{align}
%\end{comment}
%
\begin{comment}
\begin{equation}
\label{chi_f_optimal}
\begin{aligned}
\chi _{f}(G)= & \min \limits_{x_I}
&  \sum _{I\in {\mathcal {I}}(G)}x_{I} & \\
& \hspace{0.3cm}\text{s.t.}
& \sum_{I\in\mathcal{I}(G,x)} x_I &\geq 1\\
& &   x_I&\geq 0,\,\,\forall x.\\     
\end{aligned}
\end{equation}
\end{comment}
%which can be solved in the current matrix multiplication time \cite{cohen2021solving}. 
This relaxation transforms traditional coloring, which is an integer programming problem (NP-complete), into a fractional coloring problem. %(linear programming problem). %the solution to the relaxed linear program can be used to gain information about the solution to the original integer program.  
%
%
%In  the example of Fig. \ref{fig:fractional_coloring}, it is easy to note that $x_I=0.5$ for all $n(G)=5$ vertices because $|\mathcal{I}(G,x)|=2$ for all $x$.
We illustrate fractional coloring for Example \ref{uniform_example} in Fig. \ref{fig:fractional_coloring} (right). %\derya{A 6:2 coloring (middle) follows from the $3:1$ coloring.} %(left). 
A 5:2 coloring achieves the optimal solution of (\ref{chi_f_optimal}), where $|{\mathcal {I}}(G)|=10$, %with 10 independent sets in total, 
with 5 of those sets have cardinality 2, and 5 sets have size 1, and $|{\mathcal {I}}(G,x)|=3$ such that $x_I=0.5$ for sets of size 2 and $x_I=0$ for the sets with size 1. Hence, $\chi_b(\Graph)=5$ and from (\ref{fractional_graph_coloring}) $\chi_f(\Graph)=2.5$.

Example \ref{uniform_example} indicates that assigning colors to sufficiently large power graphs, we can compress $X_1$ more. A $5:2$ coloring yields $\HGchifrac=1.16$, providing a saving of $0.36$ bits over $\HGchi=1.52$. 
We sketch the traditional coloring for $\twoPowerGraph$ in Fig. \ref{fig:2nd_Power_Graph}. The distribution of colors satisfies $P(c_k)=0.16$, $k\in\{1,\dots,5\}$ and $P(c_k)=0.08$, $k\in\{6,7\}$ and $P(c_8)=0.04$. Hence, $\HGchipowertwo=1.44$. 
For a $13:2$ coloring (Fig. \ref{fig:2nd_Power_Graph_Fractional}), the colors satisfy $P(c_k)=0.08$, $k\in\{1,\dots,12\}$ and $P(c_{13})=0.04$. Hence, $\HGchipowertwofrac=0.92$, i.e., $0.52$ bits of savings from $\HGchipowertwo$. 

%Comment this paragraph in arxiv
%Due to space constraints, we next focus on the theoretical notions and the achievable gains pertaining to fractional coloring. We explain an example code construction in \cite{malak2022fractional} to demonstrate the practical savings of the proposed model.

%\begin{comment}
%%%%%
\section{Code Construction for Fractional Coloring}
\label{section:codes}

We devise binary codes for computing $f(X_1,X_2)=X_1+X_2$ in Example \ref{uniform_example} using the colorings in Figures \ref{fig:2nd_Power_Graph}-\ref{fig:2nd_Power_Graph_Fractional} to demonstrate the performance of fractional coloring. For each scenario, we consider decodings for specific  encoder realizations for $n=2$ rather than listing all $25$ possible pairs.%source combinations.

%%%
\paragraph{Traditional coloring} $n=2,\,b=1$.

{\bf Encoder.} As shown in Fig. \ref{fig:2nd_Power_Graph}, $a=8$, and the color distribution is $P(c_k)=0.16$, $k\in\{1,\dots,5\}$, $P(c_k)=0.08$, $k\in\{6,7\}$, and $P(c_8)=0.04$. For the given $\{P(c_k)\}_{k=1}^8$, the encoder devises binary %prefix free 
codewords as $Codes(2,1)=(00,011,100,101,110,111,0100,0101)$, and Kraft's inequality holds with equality. The average code length %of this scheme 
is $2.96$ bits.

{\bf Decoder.} The decoder is given the {\em Color-Function mappings}. With 2 received codes, the decoder can recover 2 function outcomes. If decoder gets the codes $00$ and $101$, it maps $00$ to color $Blue$ (which models the source subset $\{-2,2\}$ mapped to a unique function outcome) and maps $101$ to $Magenta$ (modeling $\{-2,2\}$). 
For the given joint distribution in (\ref{JointProbabilityTable}) and for computing the function  $f(X_{1i},X_{2i})=X_{1i}+X_{2i}$ for $i=1,2$, assume that the side information variable satisfies $X_{21}=-1$ and $X_{22}=1$. Then, the colors $Blue$ and $Magenta$ identify the outcomes $-2-1=-3$ and $2+1=3$, respectively. 

This encoding needs $1.48$ bits per function outcome.

%%%%%
\paragraph{Fractional coloring} $n=2,\,b=2$.

{\bf Encoder.} As shown in Fig. \ref{fig:2nd_Power_Graph_Fractional}, $a=13$. The bi-colors ($19$ unique pairs) yield 
%for the unique color pairs 
$P(c_{kl})=0.08$, $kl\in\{1,\dots,6\}$, $P(c_{kl})=0.04$, $kl\in\{7,\dots,19\}$. Given $\{P(c_{kl})\}_{kl=1}^{19}$, %the encoder devises binary length prefix free codewords are 
$Codes(2,2)=(0001,0010,0011,0101,0110,0111,1000,1001,1010,1011,\\1101,1110,1111,00001,00000,01000,01001,11000,11001)$, where Kraft's inequality holds with equality. The average code length of this scheme is $4.24$ bits.

{\bf Decoder.} The decoder needs to be told it  decodes $b=2$ colors from $n=2$ transmissions and given the ColorPair-Function mappings. Hence, the receiver has a finer-grained granularity information compared to traditional coloring. If decoder gets the codes $0001$ and $0111$, it maps $0001$ to the bi-color $BlueGreen$ (where $Blue$ models the source subset $\{-2,2\}$ mapped to a unique function outcome and $Green$ models the source subset $\{0,2\}$. Both are valid independent sets and each element in the independent set can yield the same function outcome depending on the value of $X_2$) and maps $0111$ to the bi-color $MagentaDarkgreen$ (where $Magenta$ models source subset $\{-2,2\}$ mapped to a unique function outcome and $Darkgreen$ models $\{1,-2\}$). 

With $2$ received codes and given $X_2$, the decoder identifies $2$ bi-colors, i.e., $BlueGreen$ and $MagentaDarkgreen$. For the given joint distribution in (\ref{JointProbabilityTable}) and for computing %the function 
$f(X_{1i},\,X_{2i})=X_{1i}+X_{2i}$ for $i=1,2$, assume $X_{21}=-1$ and $X_{22}=1$. Then, the bi-color $BlueGreen$ specifies the outcome pair $(-2-1,\,0-1)=(-3,\,-1)$, and the bi-color $MagentaDarkgreen$ specifies the %outcome 
pair $(2+1,\,1+1)=(3,\,2)$. Hence, the decoder can recover $4$ function outcomes. %We note that the 
Fractional coloring %argument 
requires $n\geq 2$ as it uses %the two realizations 
$X_{21}$ and $X_{22}$, which does not permit the encoder to send $2$ codewords for $n=1$.

This encoding needs $1.06$ bits per function outcome.
%\end{comment}

%%%%%
%We are now ready to present the main results of the paper. 

%I guess my approach requires a storage space due to the replicas vs the traditional coloring?

%We next detail several useful properties of traditional coloring and how it can be generalized to compute the entropy of fractional coloring. 

%%%%%%
 
%%%%%%%%%%
\section{Fractional Chromatic Entropy}
\label{section:FCE}
%The fractional graph coloring provides a natural generalization of the traditional graph coloring. %\ab{this is still not entropy}
We next formalize the notion of fractional chromatic entropy of a set of valid fractional colorings via extending Defn. \ref{chromatic_entropy}. 
\begin{defi}\label{fractional_chromatic_number}
(Fractional chromatic entropy.) 
$\coloringf$ is a valid $a:b$ fractional coloring of $\Graph$ if it assigns $b$ colors  out of a total of $a$ available colors to each $V_{X_1}$ %vertex of $\Graph$ 
such that adjacent vertices have disjoint colors.
We define $\mathcal{H}^{\chi_f}(\coloringf)=\{H(\coloringf):\, \coloringf \mbox{ is a valid a:b coloring of } G_{X_1} \}$ to be the collection of fractional chromatic entropies over the set of all valid $a:b$ colorings of $G_{X_1}$ given $X_2$.
\end{defi}
%\ab{I think you want to formally define fractional choromatic entropy first in a definition and claim it as yours. Then you'd want to say how to compute it.}\derya{i did now.}

%We next characterize the {\FCE}. %for a characteristic graph $\Graph$. 
\begin{prop}\label{FCE_characteristic_graph}
The {\FCE} of a characteristic graph $G_{X_1}$, denoted by $\HGchifrac$, is given as % by the following relation: 
\begin{align}
\label{chromatic_fractional}
\HGchifrac=\inf\limits_{b} \frac{1}{b} \min_{\coloringxf} \mathcal{H}^{\chi_f}(\coloringf)\ . 
\end{align} 
%where $\coloringf$ is a fractional coloring variable. %that assigns $b$ colors to each vertex of $\Graph$ out of $a$ available colors.
\end{prop}

\begin{proof}
%To prove this result, we use the generalized chain rule for entropy which provides a bound on the entropy rate of a collection of random variables in terms of the entropy of each variable \cite[Ch. 16.5]{cover2012elements}. %More specifically, 
Let $(Z_1,Z_2,\dots, Z_n)$ be a collection of random variables. For every $S\subseteq \{1,2,\dots,n\}$,
%why do I not use the notation $[1:n]$ here
denote by $Z(S)$ the subset $\{Z_i:\,i\in S\}$. From \cite[Ch. 16.5]{cover2012elements} the average entropy in bits per symbol of a randomly drawn $b$-element subset $Z(S)$ %of $\{Z_1,Z_2,\dots, Z_n\}$ 
is 
\begin{align}
\frac{1}{{n\choose b}} \sum\limits_{S:|S|=b} \frac{H(Z(S))}{b}\ .    
\end{align}

%https://www.ncbi.nlm.nih.gov/pmc/articles/PMC5445619/
Let $G_{X_1(S)}=\{G_{X_{1i}}:i\in S\}$ be an $b=|S|$-tuple of graphs where each $G_{X_{1i}}$ is a replica of $\Graph$. We jointly color $G_{X_1(S)}$ such that $c_{{G_{X_1(S)}}}(X_1(S))=\{c_{{G_{X_{1i}}}}(X_{1i}):i\in S\}$. Using $a$ colors in total and i.i.d. valid colorings across disjoint %sets 
$S$, the %average 
entropy %in bits per symbol 
of a randomly drawn $b$-element subset of colorings %out of $a$ available colors 
is
\begin{align}
\label{entropy_subset_full_coloring_vs_fractional_coloring}
    H(c_{{G_{X_1(S)}}}(X_1(S))\vert {\bf X}_2) = H(c^f_{{G_{X_1}}}(X_1)\vert {\bf X}_2) \ ,%/b
\end{align}
where $c^b_{{G_{X_1}}}(X_1)$ is a valid coloring of $G_{X_1(S)}$, and equivalently $c^f_{{G_{X_1}}}(X_1)$ is a valid $a:b$ coloring of $\Graph$. %and its distribution can be derived from that of $c_{{G_{X_1(S)}}}(X_1(S))$. 

The average entropy decreases monotonically in the size of the subset (Han \cite{te1978nonnegative}). As $b$ increases the rate of functional compression via fractional coloring decreases. 
%
%the minimum entropy of a coloring is called the chromatic entropy
The minimum entropy of a fractional coloring %, i.e., the {\FCE}, 
can be found by minimizing across all valid $a:b$ colorings of $\Graph$. Observing from (\ref{fractional_graph_coloring}) that
$\chi_f(\Graph)=\lim\limits_{b\to\infty}\chi_b(\Graph)/b$ and (\ref{entropy_subset_full_coloring_vs_fractional_coloring}), we obtain  (\ref{chromatic_fractional}).
\end{proof}

To visualize the rate given in Prop. \ref{FCE_characteristic_graph}, consider Example \ref{uniform_example}. %in Fig. \ref{fig:fractional_coloring}, 
The $5:2$ coloring distribution satisfies $P(c_1)=P(c_2)=2/5$ and $P(c_3)=1/5$ and $P(c_4)=P(c_5)=2/5$ ($c_3$ is repeated in $G_{X_{11}}$ and $G_{X_{12}}$). The distribution of $c_{{G_{X_1(S)}}}(X_1(S))$ across 2 graphs %satisfies $({1}/{5},{1}/{5},{1}/{5},{1}/{5},{1}/{5})$, yielding 
yields $H(c^f_{{G_{X_1}}})/2=1.16<H(\coloringx)\approx 1.52$. %\derya{A higher saving for b>2.}

\begin{comment}
For a valid $3$-fold coloring, we need $a=8$. %($7$ does not suffice because a color cannot be repeated more than $2$ times). 
The coloring distribution %for a $3$-tuple of graphs %set of graphs with size $b=3$ 
satisfies $P(c_k)=2/5$, $k\in\{1,2\}$ and $P(c_3)=1/5$ and $P(c_k)=2/5$, $k\in\{4,5\}$ ($c_3$ is repeated in $G_{X_{11}}$ and $G_{X_{12}}$), and  $P(c_k)=2/5$, $k\in\{6,7\}$, and $P(c_8)=1/5$ in $G_{X_{13}}$. The distribution of $c_{{G_{X_1(S)}}}(X_1(S))$ across 3 graphs is $({2}/{15},\dots,{2}/{15},{1}/{15})$, and $H(c^f_{{G_{X_1}}}(X_1))=0.99\times 3$.
\end{comment}

The following is an intuitive result due to a finer-grained granularity that the fractional graph coloring provides. %versus the traditional graph coloring.
Its proof follows from combining the definition in (\ref{chromatic_fractional}) and (\ref{chromatic}).
\begin{cor} The {\FCE} of a graph $G_{X_1}$ and {\CE} satisfy the following relation:
\begin{align}
\HGchifrac \leq \HGchi \ . 
\end{align}
\end{cor}
\begin{comment}
\begin{proof}
Using the definition %of {\FCE} 
in (\ref{chromatic_fractional}), it is clear that
\begin{align}
\HGchifrac\leq \min_{\coloringx}\{H(\coloring):\, \coloring \nonumber\\ \mbox{ is a valid a:1 coloring of } \Graph \vert X_2\}=\HGchi \ , \nonumber
\end{align}
where the last step follows from (\ref{chromatic}).
\end{proof}
\end{comment}

%Using K\"orner's result in \cite{korner1973coding}, (\ref{chromatic_vs_characteristic}), the following relation holds between the {\FCE} and the {\FCGE}:
Exploiting \cite{korner1973coding} the {\FCGE} satisfies
\begin{align}
\label{chromatic_vs_characteristic_fractional}
\HGfrac=\lim\limits_{n\to \infty}\frac{1}{n}\HGchipowernfrac \ ,
\end{align}
where $\chi_f$ is the fractional chromatic number of $\nPowerGraph$. 

Using (\ref{chromatic_fractional}) and (\ref{chromatic_vs_characteristic_fractional}) we can derive the {\FCGE}, a natural generalization of (\ref{chromatic_vs_characteristic}). %the {\CGE}. 
Prop. \ref{FCE_characteristic_graph_n_limit} is derived from Prop. \ref{FCE_characteristic_graph} and (\ref{chromatic_vs_characteristic_fractional}), and we skip its proof.

\begin{prop}
\label{FCE_characteristic_graph_n_limit}
The {\FCGE} is given as
\begin{multline}
\label{FCGE}
\HGfrac=\lim\limits_{n\to \infty}\frac{1}{n} \inf\limits_{b} \frac{1}{b} \min\nolimits_{\coloringpowernxf}\{H(\coloringpowernf):\, \\ 
\coloringpowernf \mbox{ is a valid a:b coloring of } \nPowerGraph \vert\, {\bf X}_2^n\} \ ,   
\end{multline}
where $\coloringpowernf$ is a fractional coloring variable that assigns $b$ colors to each vertex of $\nPowerGraph$ out of $a\ge b$ available colors.
\end{prop}

\begin{lem}\label{FCGEvsCGE}
%\ab{I'd say this is a lemma. Wouldn't the RHS be constrained to b=1 and the LHS does an inf over $b \in \mathbb{N}$ in which case this would be obvious?}
The following relation holds for the {\FCGE} and {\CGE}:
\begin{align}
\HGfrac \leq \HG\ . %why not \leq ????
\end{align}
\end{lem}

\begin{proof}
We use an important result that ties the {\FCN} to the n-th power of a graph.

The following relation between the n-th power of $G$ and the {\FCN} holds \cite[Corollary 3.4.3]{scheinerman2011fractional}:
\begin{align}
\label{fractional_vs_regular_chromatic number}
\chi_f(G)=\inf\limits_{n} \sqrt[n]{\chi(G^n)}=\lim\limits_{n\to\infty} \sqrt[n]{\chi(G^n)}\ .    
\end{align}
The relation (\ref{fractional_vs_regular_chromatic number}) implies that $\chi(G^n) \approx \chi_f(G)^n$ as $n\to\infty$.

It also holds from \cite[Corollary 3.4.2]{scheinerman2011fractional} that 
\begin{align}
\label{fractional_product_property}
\chi_f (G^n)=\chi_f(G)^n \ .    
\end{align}

%For the proofs of (\ref{fractional_vs_regular_chromatic number}) and (\ref{fractional_product_property}), the reader is referred to Theorems 1.6.1 and 1.6.2 in \cite{scheinerman2011fractional}.

As a result, we infer for the m-th power of $G$ that $\chi_f(G^m)\overset{(\ref{fractional_product_property})}{=} \chi_f(G)^m \overset{(\ref{fractional_vs_regular_chromatic number})}{=} (\lim\limits_{n\to\infty} \sqrt[n]{\chi(G^n)})^m$, and
$\chi_f(G^m)\overset{(\ref{fractional_vs_regular_chromatic number})}{=}\lim\limits_{n\to\infty} \sqrt[n]{\chi(G^{n\cdot m})}%\nonumber\\
%&
\leq %\lim\limits_{n\to\infty} \sqrt[n]{\chi(G^{n})^m}=
\lim\limits_{n\to\infty} \chi(G^{n})^{\frac{m}{n}}$. 
The fractional coloring requires $\log \chi_f(G^n)%=n\log \chi_f(G)
$ bits which is less than $\log \chi(G^n)$ bits as required by the traditional coloring. 
\end{proof}

The current paper aims to improve the compression rate by introducing fractional chromatic entropy. %We emphasize that our definition does not conflict with the independent set-based definition 
On the other hand, this approach does not outperform the independent set-based fundamental limit for graph entropy that establishes the optimal rate for lossless function computation $f(X_1,X_2)$ given side information $X_2$ \cite[Theorem 21.2]{el2011network}. %which is given by the graph entropy in \cite[Theorem 21.2]{el2011network}. 

%%%%%%%%%%%%%%%%%%%%%%%%%%%%%%%%%%%%%%%
\section{Coding Gains of Fractional Coloring}
\label{section:coding_gains}

We denote the integrality gap (IG), i.e., ratio of the solutions of the traditional coloring versus the fractional coloring problems for encoding $\nPowerGraph$, by $IG_n$. It is given as
%Approximation and integrality gap %https://en.wikipedia.org/wiki/Linear_programming_relaxation
\begin{align}
IG_n = \frac{\HGchipowern}{\HGchipowernfrac} \ ,\quad n\geq 1\ .\nonumber
%\frac{\frac{1}{n}\HGchipowern}{\frac{1}{n}\HGchipowernfrac}
\end{align}
From Lemma \ref{FCGEvsCGE}, it is immediate that $IG_n \geq 1$ for $n\geq 1$.

\begin{lem}\label{IG_as_function_n}
%\ab{I'd consider this a lemma}\derya{not sure...}
$IG_n$ is an increasing function of $n$.
\end{lem}

\begin{proof}

We can rewrite $IG_n$ as $IG_n={\frac{1}{n}\sum\limits_{k=1}^n a_k}\,\Big/\,{\frac{1}{n}\sum\limits_{k=1}^n b_k}$, where $a_k=\HGchipowerk-\HGchipowerkminusone$ and $b_k=\HGchipowerkfrac-\HGchipowerkminusonefrac$ which both decrease in $k$ because the sequences $\HGchipowerk$ and $\HGchipowerkfrac$ are increasing and concave in the sense of decreasing slope, which can be shown using Han's theorem \cite{te1978nonnegative}. Furthermore, $a_k\geq b_k$ for $k\in\{1,2,\dots,n\}$. Due to fractional coloring $b_k$ decreases with a higher rate versus $a_k$. We infer that $1\leq \frac{a_1}{b_1}\leq \frac{a_2}{b_2}\leq \dots \frac{a_n}{b_n}$. Using this relation, we show that 
\begin{align}
\frac{a_1}{b_1}\leq \frac{a_1+a_2}{b_1+b_2}\leq \dots \leq \frac{a_1+a_2+\dots+a_n}{b_1+b_2+\dots+b_n}=IG_n\ .\nonumber
\end{align}
Hence, $IG_n$ is an increasing function of $n$.
\end{proof}

%our running example 
In Example \ref{uniform_example}, $IG_1=\frac{1.52}{1.16}=1.31$, $IG_2=\frac{1.44}{0.92}=1.57$, and $IG_n$ is higher for $n>2$ (Lemma \ref{IG_as_function_n}). %We expect a higher IG for $n>2$. 
The gain is %mainly 
due to cross-coding across graphs (Prop. \ref{FCE_characteristic_graph}). For the identity function, the $\Graph$ %characteristic graph 
is complete, $b^*=1$, and fractional coloring does not have savings. Significant gains are possible for sparse graphs. %(graphs with only a few edges)

For a valid a:b coloring of $\Graph$, let $b^*_{\Graph}$ be the smallest $b=|S|$ that achieves (\ref{chromatic_fractional}) for $\Graph$, and $\coloring$ and $\coloringf$ be the valid colorings with distributions ${\bf q}=(q_1,\dots,q_{\chi(\Graph)})$ and ${\bf r}=(r_1,\dots,r_{\chi_{b^*_{\Graph}}(G_{X_1(S)})})$ that minimize the respective entropies of the colorings. 
Similarly, for $n>1$, let $b^*_{\nPowerGraph}$ be the smallest $b=|S|$ for $\nPowerGraph$ satisfying (\ref{fractional_graph_coloring}),  $\coloringpowern$ and $\coloringpowernf$ be the valid colorings with %distributions 
${\bf q}^n=(q_1,\dots,q_{\chi(\nPowerGraph)})$, ${\bf r}^n=(r_1,\dots,r_{\chi_{b^*_{\nPowerGraph}}(G^n_{{\bf X}_1(S)})})$. %that minimize the respective entropies. 
%We next lower bound $IG_n$. 

\begin{prop}\label{prop_IG_lower_bound}
The fractional coloring scheme attains
\begin{align}
%\label{IG_lower_bound}
IG_n\geq 
\frac{b^*_{\nPowerGraph}H({\bf q}^n)}{H({\bf q}^n)+\Delta_{G^n}} \ ,\nonumber
\end{align}
where
\begin{align}
\Delta_{G^n}\!=\!\!\!\sum\limits_{j\in \mathcal{J}_{G^n}}\!\!\!q_j \Big[h\Big(\frac{1}{b^*_{\nPowerGraph}}\Big)\!+\!\frac{b^*_{\nPowerGraph}\!\!-\!1}{b^*_{\nPowerGraph}}\!\log(m_{G^n}(j)(b^*_{\nPowerGraph}\!-\!1))\Big] \ \!\!,  \nonumber 
\end{align}
$j\in \mathcal{J}_{G^n}=\{1,\dots,\chi_b(G^n_{{\bf X}_1(S)})\}$ represents the coloring class, %for traditional coloring,
and $m_{G^n}(j)$ is the count of class $j$ vertices. 
\end{prop}

\begin{proof}
We first focus on $n=1$. Then,
\begin{align}
IG_1=b^*_{\Graph}\frac{H(\coloring\vert X_2)}{H(\coloringf\vert X_2)}
=b^*_{\Graph}\frac{H({\bf q})}{H({\bf r})}\ .\nonumber
\end{align}
Using the grouping property of entropy \cite[Ch.2 ]{cover2012elements}, 
\begin{align}
\label{grouping}
H({\bf r})=H({\bf q})+\sum_{j}q_j H\Big((\frac{r_l}{q_j}:\, \sum\nolimits_{l\in j} r_l=q_j)\Big) \ .  
\end{align}
We next analyze the RHS of (\ref{grouping}) for a given $j$, assuming that there exist $m_{G}(j)$ vertices in $\Graph$ in color class $j$. It then holds that in the $b$-fold coloring scheme $q_j$ accumulates the probabilities of $m_{G}(j)\times b$ vertices in $G_{X_1(S)}$. The marginal distribution of colors in each $G_{X_{1i}}:\, i\in S$ is identical. Hence,
\begin{multline}
H\Big((\frac{r_l}{q_j}:\, \sum\limits_{l\in j} r_l=q_j)\Big) \leq H\Big(\frac{1}{m_{G}(j) b},\dots, \frac{1}{m_{G}(j) b},\frac{1}{b}\Big)   \nonumber\\
=h\Big(\frac{1}{b}\Big)+\frac{b-1}{b}\log(m_{G}(j)(b-1))\ ,\nonumber
\end{multline}
where the colors of $G_{X_{11}}$ and $\Graph$ are the same, putting $1/b$ of the mass of $q_i$ in $G_{X_1(S)}$, and leaving $(b-1)/b$ of the mass to the remaining $b-1$ graph replicas $\{G_{X_{1i}}:\, i\in S,\, i\neq 1\}$ with $m_{G}(j)(b-1)$ vertices. The RHS holds when the colors are uniformly split among $m_{G}(j)(b-1)$ vertices. Hence,
\begin{align}
%\label{IG_lower_bound}
IG_1\!\geq \!\frac{b^*_{\Graph}H({\bf q})}{H({\bf q})+\Delta_G} \ ,\nonumber
\end{align}
where $\Delta_G=\!\!\sum\limits_{j\in\mathcal{J}_G}\!\!\!q_j \Big[h\Big(\frac{1}{b^*_{\Graph}}\Big)\!\!+\!\frac{b^*_{\Graph}\!\!-\!1}{b^*_{\Graph}}\!\log(m_{G}(j)(b^*_{\Graph}\!\!\!-\!1))\Big]$. 
For $n>1$, employing the grouping property we can obtain the final result. We note that the count of class $j$ vertices $m_{G^n}(j)$ is less than $m_{G}(j)^n$ for $j\in \mathcal{J}_G=\{1,\dots,\chi(G_{X_1})\}$.
%\begin{align}
%IG_n =b^*_{\nPowerGraph}\frac{H(\coloringpowern\,\vert\,{\bf X}_2)}{H(\coloringpowernf\,\vert\,{\bf X}_2)}=b^*_{\nPowerGraph}\frac{H({\bf q}^n)}{H({\bf r}^n)}\ .
%\end{align}
\end{proof}

From Prop. \ref{prop_IG_lower_bound}, for any given $\Graph$, if not a complete graph, fractional coloring for $n>1$ offers compression savings.

The following corollary implies theoretically that the achievable $IG_n$ is lower bounded by $b^*_{\nPowerGraph}$ (up to a scaling), for a valid $b^*_{\nPowerGraph}$-fold coloring that is determined by the function to be computed and its n-th power graph $\nPowerGraph$. 
\begin{cor}\label{cor_IG_lower_bound}
Under the assumption that the colorings of $\nPowerGraph$ are uniform, the fractional coloring scheme attains
\begin{align}
IG_n\geq b^*_{\nPowerGraph}\cdot {\log\chi_f(\Graph)}\,\Big/\,{\log\chi_{b^*_{\Graph}}(\Graph)} \ .\nonumber
\end{align}
\end{cor}
\begin{proof}
Let $b^*_{\Graph}$ and $b^*_{\nPowerGraph}$  be the smallest $b$ values such that (\ref{fractional_graph_coloring}) holds for $\Graph$ and $\nPowerGraph$, respectively. Then, provided that the colorings of $\nPowerGraph$ are uniform, we can simplify the expressions for $H({\bf q})$ and $H({\bf r})$ and obtain
\begin{align}
\label{IGn_LB1}
IG_n 
&\geq 
\frac{ \frac{1}{n}\log\chi_f(\Graph)^n}{\frac{1}{n}\cdot\frac{1}{b^*_{\nPowerGraph}}\cdot\log (b^*_{\nPowerGraph}\cdot\chi_f(\nPowerGraph))} \\
%&\overset{(\ref{fractional_product_property})}{=}\frac{n\log\chi_f(\Graph)}{\frac{1}{b^*_{\nPowerGraph}}\cdot\log ({b^*_{\nPowerGraph}\cdot\chi_f(\Graph)^n)}}\nonumber\\
\label{IGn_LB2}
&=b^*_{\nPowerGraph}\cdot\frac{\log\chi_f(\Graph)}{\log {(b^*_{\nPowerGraph})^{1/n}}+\log\chi_f(\Graph)}\\
\label{IGn_LB3}
&\geq b^*_{\nPowerGraph}\cdot {\log\chi_f(\Graph)}\,\Big/\,{\log\chi_{b^*_{\Graph}}(\Graph)}\ ,\quad\,%\qedhere
\end{align}
where the lower bound (\ref{IGn_LB1}) follows from using  %(\ref{fractional_graph_coloring}) and 
(\ref{fractional_vs_regular_chromatic number}), and (\ref{IGn_LB2}) is due to (\ref{fractional_product_property}). 
Employing the relations $\chi_f(\Graph)=\frac{\chi_{b}(\Graph)}{b}$,  (\ref{fractional_product_property}), which yields $\chi_f(\nPowerGraph)=\frac{\chi^n_{b}(\Graph)}{b^n}$, and $\chi^n_{b}(\Graph)\geq \chi_{b}(\nPowerGraph)$, we obtain $b^*_{\nPowerGraph}\leq (b^*_{\Graph})^n$, %. This gives the last inequality 
which yields (\ref{IGn_LB3}).
\end{proof}

%%%%%%%%%%%%%%%%%%%%%%%%%%%%%%%%%%%%%%%
{\bf \em Discussion and future directions.} Fractional coloring exploits the possibility of cross-coding between graphs and provides coding gains (at infinite and finite source sequence lengths). This approach provides a reduced communication complexity versus traditional coloring via decreasing the number of bits to send roughly from $\log \chi$ to $\log \chi_f$. Lower bounding $IG_n$ for different classes of functions and source distributions is of primary importance. Quantifying and upper bounding the IG in the limit of large $n$, i.e., the ratio of the bits required by fractional compression, $\HGfrac$, to $\HG$, is left as future work.

%There are several implementation bottlenecks. First, 
While fractional coloring is a less combinatorial problem than traditional coloring and accepts a linear programming solution (solvable in polynomial time), finding an independent set is strongly NP-hard \cite{garey1978strong}. Hence, the relaxation is in the class of NP-hard problems. %\cite[Ch. 3.9, 4.5]{scheinerman2011fractional} %3.9 Computational complexity %Ironically, this negative result follows from the positive result that the ellipsoid algorithm for solving linear programs does run in polynomial time. It turns out that the ellipsoid algorithm gives a polynomial transformation between the fractional coloring problem and the problem of computing the independence number of a graph. The latter invariant is known to be NP-hard to compute.
This issue can be alleviated via using {\em fractional edge coloring} (versus {\em fractional vertex coloring}) for which a polynomial-time solution exists \cite{scheinerman2011fractional}. %[Theorem 4.2.1]
%\cite{scheinerman2011fractional}Edge coloring : The minimum required number of colors for the edges of a given graph is called the chromatic index (or edge chromatic number, $\chi′(G)$.) of the graph.
%The computation of the chromatic number and edge chromatic number are NP-hard [117]. As noted in the previous chapter, it is also NP-hard to compute the fractional chromatic number of a graph. Thus, at first glance, one might expect that the fractional edge chromatic number is just as intractable. However, Theorem 4.2.1 on page 59 expressed the fractional edge chromatic number (a minimization problem) in terms of $\Delta(G)$ and $\Lambda(G)$ (which are maximization problems). A consequence of this minimax theorem is that the problem ``Given a graph G and integers a, b > 0, decide if $\chi'f (G) \leq a/b$" is in NP and co-NP. This suggests that a polynomial-time solution ought to exist, and indeed one does.
\begin{comment}
Next, fractional coloring in the limit of large $b$ may not be feasible %due to a finer quantization (graph-quantization) requirement
in practice, and {\em rounding or binning schemes} may be used post-fractional coloring for retaining the typical sequences in a sufficiently large power graph. %fractional coloring followed by rounding.
\end{comment}
The characteristic graph approach is concerned with functional compression of source sequences when the adjacency matrix is a (0,1)-matrix. A possible generalization includes edge-weighted graphs %where the adjacency matrix stores edge weights 
to capture the distortion in reconstruction.%for the computation tasks.

%\section*{Acknowledgment}

%We are indebted to ... for helpful discussions and the feedback on the paper. 

%%%%%%
%\IEEEtriggeratref{3}
\Urlmuskip=0mu plus 1mu\relax
\bibliographystyle{IEEEtran}
\bibliography{ref}

\newpage

\end{document}